\documentclass[final,3p,times]{elsarticle}

\usepackage{amssymb}
\usepackage{amsthm}
\usepackage{amsmath}

\journal{Journal of Computing and System Sciences}

\theoremstyle{plain}
\newtheorem{theorem}{Theorem}

\newtheorem{lemma}{Lemma}
\newtheorem{proposition}{Proposition}

\theoremstyle{definition}
\newtheorem{definition}{Definition}

\def\psigam#1{{\psi(#1,\,\gamma(#1))}}

\begin{document}

\begin{frontmatter}


\title{Almost-Natural Proofs}
\author{Timothy Y. Chow}
\ead{tchow@alum.mit.edu}
\address{Center for Communications Research, 805 Bunn Drive, Princeton, NJ
08540}


\begin{abstract}
Razborov and Rudich have shown that so-called
\emph{natural proofs} are not useful for separating $P$ from~$NP$
unless hard pseudorandom number generators do not exist.
This famous result is widely regarded as a serious barrier to proving
strong lower bounds in circuit complexity theory.

By definition, a natural combinatorial property satisfies two
conditions, \emph{constructivity} and \emph{largeness}.
Our main result is that
if the largeness condition is weakened slightly,
then not only does the Razborov--Rudich proof break down,
but such ``almost-natural'' (and useful) properties provably exist.
Specifically, under the same pseudorandomness assumption
that Razborov and Rudich make,
a simple, explicit property that we call \emph{discrimination}
suffices to separate $P/poly$ from~$NP$;
discrimination is nearly linear-time computable and almost large,
having density $2^{-q(n)}$ where $q$ is a quasi-polynomial function.
(This is a slightly stronger result than the one announced
in the FOCS 2008 extended abstract of this paper.)
For those who hope to separate $P$ from~$NP$ using
random function properties in some sense,
discrimination is interesting, because it is constructive,
yet may be thought of as a minor alteration of
a property of a random function.

The proof relies heavily on
the self-defeating character of natural proofs.
Our proof technique
also yields an \emph{unconditional} result,
namely that there exist almost-large and useful properties
that are constructive, if we are allowed to call
\emph{non-uniform} low-complexity classes ``constructive.''
We note, though, that this unconditional result can also be proved
by a more conventional counting argument.

Finally, we give an alternative proof
(communicated to us by Salil Vadhan at FOCS 2008) of one of our theorems,
and we make some speculative remarks on the future prospects for
proving strong circuit lower bounds.
\end{abstract}


\end{frontmatter}


\section{Introduction}
\label{sec:intro}

In a famous paper~\cite{Razborov-Rudich},
Razborov and Rudich introduced the concept
of a \emph{natural combinatorial property} of a Boolean function.
They showed on the one hand
that almost all lower bounds in circuit complexity theory
proved up to that time
(specifically, all non-relativizing, non-monotone, superlinear lower bounds)
had employed natural properties,
and on the other hand that natural properties cannot be used
to separate $P$ from $NP$ unless
$2^{n^\epsilon}$-hard pseudorandom number generators do not exist.
Their result is widely regarded as a serious barrier to proving
strong circuit lower bounds.

In more detail, if $\Gamma$ and $\Lambda$ are complexity classes,
then Razborov and Rudich say that
a property of Boolean function on $n$ variables is
\emph{$\Gamma$-natural of density~$\delta_n$
and useful against~$\Lambda$}
if (roughly speaking) the property is $\Gamma$-computable
(from the truth table of a given Boolean function),
if it holds for $2^{2^n} \delta_n$ Boolean functions,
and if it contains no $\Lambda$-computable Boolean functions.
They showed that if $\Gamma=\Lambda= P/poly$
and $\delta_n = \Omega(2^{-poly(n)})$,
then no such properties exist unless
$2^{n^\epsilon}$-hard pseudorandom number generators do not exist.
Informally, if a property is \emph{constructive}
($\Gamma$ is sufficiently weak) and \emph{large}
($\delta_n$ is sufficiently large),
then it is not likely to be useful for proving strong circuit lower bounds.

It follows that if we believe in hard pseudorandom number generators
but still wish to prove circuit lower bounds,
then we are led to ask just \emph{how} non-constructive and/or small
a property needs to be in order to circumvent the
so-called ``naturalization barrier.''
Rudich~\cite{Rudich} has shown that if
we allow ourselves to assume a stronger pseudorandomness hypothesis,
then the naturalization barrier remains intact even if
constructivity is weakened to $N\tilde P/qpoly$-constructivity.
On the other hand, as pointed out by a referee of
the FOCS 2008 extended abstract of this paper,
for any fixed~$k$ there are properties computable in time $2^{n^{k+1}}$
that are useful against circuits of size $n^k$
(simply use brute-force search).

In this paper we investigate the weakening of the largeness condition.
The main result is that under the same
$2^{n^\epsilon}$-hard pseudorandomness assumption of
the original Razborov--Rudich paper,
we can explicitly exhibit a nearly-linear-natural
property that separates $NP$ from $P/poly$ and whose density is $2^{-q(n)}$
where $q$ is a quasi-polynomial function
(whose degree depends on $\epsilon$ and on the
size of the pseudorandom number generator).
Of course, the pseudorandomness hypothesis trivially implies
the existence of constructive properties
that separate $NP$ from $P/poly$;
for example, simply take an explicit family of
$NP$-complete Boolean functions.
However, this latter family has density $2^{-e(n)}$
for some function $e(n)$ that grows exponentially;
this is far smaller than~$2^{-q(n)}$.

The main idea of our proof is to exploit
the \emph{self-defeating nature} of natural proofs.
Assume that natural, useful properties do not exist
(for example, by assuming
that $2^{n^\epsilon}$-hard pseudorandom number generators
exist and invoking Razborov--Rudich).
This means that every attempt to find a natural property
that discriminates nonconstructive functions from constructive ones fails.
The key observation is that
\emph{a natural property
is itself just a constructive function}
(a constructive function of a truth table, that is,
but a truth table is just an arbitrary binary string).
Therefore we have identified a feature
that every constructive function has:
It is no good at discriminating nonconstructive functions from
constructive ones.
So if we consider the property of \emph{discrimination},
i.e., the ability to distinguish nonconstructive functions from
constructive ones,
then \emph{discrimination is a useful property.}
On the other hand,
it is easy to prove unconditionally that
discrimination is almost large,
and that discrimination is constructive.\footnote{Note
that even if every \emph{discriminating function} is non-constructive,
the \emph{property of discrimination} is constructive,
because it is easy to check, given the truth table of a function~$f$,
whether $f$ is a discriminating function.}
Moreover, one can explicitly describe an $NP$ function
that is discriminating,
so discrimination separates $NP$ from $P/poly$.
This is our main result.

The key point is that it is the very assumption that natural, useful
properties do not exist that yields a useful property.

One can ask whether the above line of reasoning can be used to
prove an \emph{unconditional} result,
just as Avi Wigderson adapted Razborov and Rudich's argument
to prove unconditionally that there is no natural proof
that the discrete logarithm problem is hard.
Indeed, this is possible, as we show below.
It turns out, however, that this unconditional result
can be proved using a direct counting argument.

We hope that these results will give some insight into
how to bypass the naturalization barrier.
If $2^{n^\epsilon}$-hard pseudorandom number generators do not exist,
then of course the naturalization barrier evaporates.
On the other hand, if such generators \emph{do} exist,
then our results show that there exists at least
one property (namely, discrimination) that separates $NP$ from $P/poly$
and that is both constructive and---as we shall see shortly---only
a minor alteration of a random property.

\section{Preliminaries}
\label{sec:background}

We write $\mathbb{N}$ for the positive integers,
and our logarithms are always base~$2$.
All gates in our Boolean circuits are assumed to have just two inputs.
We use the notation $(x_n)$ to denote a sequence $x_1, x_2, \ldots\,$,
and whenever we refer to a sequence $(f_n)$ of Boolean functions,
we always understand that $f_n$ is a function of~$n$ variables.
Given a function $\lambda:\mathbb{N}\to\mathbb{N}$,
we write $SIZE(\lambda)$ to denote the
complexity class comprising all sequences $(f_n)$
of Boolean functions for which there exists a constant~$c$
such that the minimum circuit size of~$f_n$ is at most $c\lambda(n)$
for all sufficiently large~$n$.
The following definition will also be convenient.
\begin{definition}
Given two functions $\gamma:\mathbb{N}\to\mathbb{N}$
and $\lambda:\mathbb{N}\to\mathbb{N}$,
we say that \emph{$\gamma$ outstrips~$\lambda$}
if for every constant $c>0$ there exists $n_0$ such
that $\gamma(n) > c\lambda(n)$ for all $n\ge n_0$.
That is, $\gamma$ eventually grows \emph{strictly} faster than
any constant times~$\lambda$.
\end{definition}

Now let us review some fundamental concepts from~\cite{Razborov-Rudich}.

\begin{definition}
A \emph{Boolean function property} (or just \emph{property} for short)
is a sequence $C = (C_n)$ where each $C_n$ is a set of
Boolean functions on $n$ variables.
\end{definition}

\begin{definition}
If $\Gamma$ is a complexity class and
$(\delta_n)$ is a sequence of positive real numbers,
then a property $(C_n)$ is
\emph{$\Gamma$-natural with density~$\delta_n$} if
\begin{enumerate}
\item (largeness) $|C_n| \ge 2^{2^n} \delta_n $
for all sufficiently large~$n$;
and
\item (constructivity)
the problem of determining whether $f_n \in C_n$,
given as input the full truth table of
a Boolean function $f_n$ on $n$ variables,
is computable in~$\Gamma$.
\end{enumerate}
\end{definition}

Note that our definition of \emph{natural}
differs slightly from that of Razborov and Rudich;
for them, a natural property is one which \emph{contains}
a large and constructive property.
This difference will do no harm,
because our results assert
the \emph{existence} of certain natural properties in our sense,
and a property that is natural in our sense is also natural in
Razborov and Rudich's sense.

Later on we will be particularly interested in the case
of \emph{nearly-linear-natural} properties,
which we define to mean $\Gamma = DTIME(N (\log N)^c)$
for some constant~$c$.
Here we have used an uppercase~$N$
to emphasize that ``nearly linear''
means nearly linear in~$N = 2^n$, the size of the truth table of~$f_n$.

Next we recall the definition of a \emph{useful} property.

\begin{definition}
If $\Lambda$ is a complexity class, then
a property $(C_n)$ is \emph{useful against $\Lambda$}
if for every sequence $(f_n)$ of Boolean functions satisfying
$f_n \in C_n$ for infinitely many~$n$, $(f_n) \notin \Lambda$.
\end{definition}

For our purposes we also need a slightly weaker notion,
which we shall call \emph{quasi-usefulness}.

\begin{definition}
If $\Lambda$ is a complexity class, then
a property $(C_n)$ is \emph{quasi-useful against $\Lambda$}
if for every sequence $(f_n)$ of Boolean functions satisfying
$f_n \in C_n$ for all sufficiently large~$n$, $(f_n) \notin \Lambda$.
\end{definition}

The difference between usefulness and quasi-usefulness is that
there may be infinitely many~$n$
for which a quasi-useful property
is easy to compute, whereas this cannot happen for a useful
property.\footnote{As pointed out by a referee,
our distinction between \emph{useful} and \emph{quasi-useful}
is the same as the distinction
between \emph{diagonalization a.e.}\ and \emph{diagonalization i.o.}\
in~\cite{Regan-Sivakumar-Cai}.}
However, a quasi-useful property retains the important characteristic
of not containing any $\Lambda$-computable sequence of Boolean functions.
So for the purpose of separating $\Lambda$ from a higher complexity class,
quasi-usefulness suffices.

Note that the only reason we introduce quasi-usefulness is
to handle the slightly annoying technicality that
the length of a truth table is not an arbitrary integer
but is always a power of two.
An alternative way around this technicality might be to
pad out strings whose lengths are not powers of two.

\begin{definition}
Fix $\epsilon>0$.
A family of functions $G_n:\{0,1\}^n \to \{0,1\}^{2n}$ is a
\emph{$2^{n^\epsilon}$-hard pseudorandom number generator} if
for every circuit~$C$ with fewer than $2^{n^\epsilon}$ gates,
\[
\left| \mathrm{Prob}[C(G_n(\mathbf{x})) = 1] - 
  \mathrm{Prob}[C(\mathbf{y}) = 1] \right| < 1/2^{n^\epsilon}.
\]
Here $\mathbf{x}$ is chosen at random from $\{0,1\}^n$
and $\mathbf{y}$ is chosen at random from $\{0,1\}^{2n}$.
\end{definition}

We are now ready for Razborov and Rudich's fundamental result.
We need a slightly stronger version of the theorem than
the one that appears in their paper.

\begin{theorem}[Razborov--Rudich]
\label{thm:RR}
Fix $c\ge 1$, $d>1$, and $\epsilon > 0$.
Assume that there exists a $2^{k^\epsilon}$-hard
pseudorandom number generator~$G_k$ in $SIZE(k^{c})$.
Then for any $e> 1 + cd/\epsilon$,
there is no $P/poly$-natural property
with density greater than~$2^{-n^d}$
that is useful against $SIZE(n^e)$.
\end{theorem}

\begin{proof}
Only minor changes to Razborov and Rudich's argument are needed,
but for completeness we give a full proof.

Choose any $e > 1 + cd/\epsilon$.
We use our pseudorandom \emph{number} generator~$G$
to construct a pseudorandom \emph{function} generator~$f$.
For every $k\ge 1$,
let $G_k^0, G_k^1 : \{0,1\}^k \to \{0,1\}^k$ be the first and last $k$ bits
of $G_k$ respectively.
For the rest of the proof,
we will write $n$ for $\lfloor k^{\epsilon/d}\!/2 \rfloor$.
For any $k$-bit string~$x$,
let $f(x)$ be the Boolean function that sends $y \in \{0,1\}^n$ to
the first bit of
\begin{equation*}
G_k^{y_n} \circ G_k^{y_{n-1}} \circ \cdots \circ G_k^{y_1}(x).
\end{equation*}
We claim that the family of functions $\{f(x)\}$ is in $SIZE(n^e)$.
This is because
$G_k$ is in $SIZE(k^c) \subseteq SIZE(n^{e-1})$,
and it is straightforward to build a circuit for $f(x)$
using $n$ copies of~$G_k$ (with the $i$th bit of the input dictating
which half of the $i$th copy of $G_k$ to use).

Now assume towards a contradiction that there exists
a $P/poly$-natural property $(C_n)$ with density at least~$2^{-n^d}$
that is useful against $SIZE(n^e)$.
Then for all sufficiently large~$k$,
none of the functions $f(x)$ are in~$C_n$.
Therefore if $\mathbf{f}_n$ denotes a randomly chosen
Boolean function on $n$ variables and
$\mathbf{x}$ denotes a randomly chosen $k$-bit string, then
\begin{equation}
\label{eq:randtest}
\left| \mathrm{Prob}[C_n(f(\mathbf{x})) = 1] - 
  \mathrm{Prob}[C_n(\mathbf{f}_n) = 1] \right| \ge 2^{-n^d}.
\end{equation}

Equation~(\ref{eq:randtest})
gives us a statistical test for $f(\mathbf{x})$,
which we now convert into a statistical test for~$G_k$.
Let $T$ be a binary tree of height~$n$,
having $2^n-1$ internal nodes and $2^n$ leaves.
Construct a labeling $\ell$ of the nodes of~$T$
by labeling the leaves with (distinct) $n$-bit binary strings
and labeling the internal nodes with (distinct) numbers $1$ to~$2^n-1$
in such a way that
if $u$ and~$v$ are internal nodes and $u$ is a child of~$v$,
then $\ell(u)<\ell(v)$.
If $y$ is a leaf of~$T$,
then let $\ell(y)(j)$ denote the $j$th bit of~$\ell(y)$.
For $i\in\{0,1,\ldots, 2^n-1\}$,
let $T_i$ be the subforest of~$T$ induced by
the set of internal nodes~$v$ with $\ell(v)\le i$,
together with all the leaves.
If $y$ is a leaf of~$T$, then let $v_i(y)$ be
the root of the subtree of~$T_i$ containing~$y$,
and let $h(i,y)$ be the distance between $v_i(y)$ and~$y$
(so for example $h(i,y)=0$ if $v_i(y) = y$).

Now define independent random variables $\mathbf{x}(v)$,
one for each node $v$ of~$T$,
and each chosen uniformly from $\{0,1\}^k$.
Define a random collection $\mathbf{f}_{i,n}$ by
letting $\mathbf{f}_{i,n}(y)$ (for a leaf~$y$ of~$T$) be the first bit of
\begin{equation*}
G_k^{\ell(y)(n)} \circ G_k^{\ell(y)(n-1)} \circ \cdots \circ
  G_k^{\ell(y)(n - h(i,y) + 1)}(\mathbf{x}(v_i(y))).
\end{equation*}
Then $\mathbf{f}_{0,n}$ is $\mathbf{f}_n$ and
$\mathbf{f}_{2^n-1,n}$ is $f(\mathbf{x})$,
so Equation~(\ref{eq:randtest}) implies that for some~$i$,
\begin{equation}
\label{eq:bias}
\left| \mathrm{Prob}[C_n(\mathbf{f}_{i-1,n})) = 1] - 
  \mathrm{Prob}[C_n(\mathbf{f}_{i,n}) = 1] \right|
  \ge 2^{-n^d}\!/2^n \ge 2^{-2n^d},
\end{equation}
since $d\ge 1$.
There must exist some assignment of the values of the $\mathbf{x}(v)$
for all roots~$v$ of subtrees of $T_i$
except the root $u$ with $\ell(u)=i$,
such that Equation~(\ref{eq:bias}) still holds
when conditioned on this assignment.
By fixing such an assignment, we obtain
a statistical test that distinguishes between
$G_k(\mathbf{x}(u))$ and $(\mathbf{x}(u'), \mathbf{x}(u''))$,
where $u'$ and $u''$ are the children of~$u$,
and that can be computed by circuits of size $2^{O(n)}$
(because $(C_n)\in P/poly$).
But this contradicts the $2^{k^\epsilon}$-hardness of~$G_k$,
because for all sufficiently large~$k$,
$k^\epsilon$ is larger than $2n^d$ and also larger
than any constant times~$n$.
\end{proof}

Finally, we need some estimates for the size of $\psi(n, g)$,
the number of Boolean functions of $n$ variables
that can be computed by Boolean circuits with at most $g$~gates.
The upper bound is due essentially to Shannon,
though the version we quote here is Lemma~2.1
in~\cite{Savicky-Woods}.

\begin{proposition}
\label{prop:shannon}
For all $n\ge1$ and $g\ge 1$,
$\psi(n, g) < g^g e^{g+4n}$.
\end{proposition}

For the proof of Theorem~\ref{thm:counting},
we will also need a lower bound on $\psi(n,g)$.
This result is somewhat technical and is not needed
for the proofs of Theorem~\ref{thm:main} or Theorem~\ref{thm:nonunif},
so the reader can skip to Section~\ref{sec:main} now
without loss of continuity, returning to the lemmas below when needed.

We need a couple of facts about binomial coefficients.
These facts are well known to experts, but for completeness
we give the proofs.
The first fact is an elementary large-deviation result.

\begin{lemma}
\label{lem:largedev}
If $k\le (1/2 - \epsilon)N$, then there is a constant $c > 0$
(depending on $\epsilon$ but not on $N$ or~$k$) such that
\begin{equation}
\label{eq:largedev}
\sum_{i=0}^k \binom{N}{i} \le c \binom{N}{k}.
\end{equation}
\end{lemma}

\begin{proof}
Let $S$ denote the sum on the left-hand side of~(\ref{eq:largedev}).
The ratio between consecutive terms in~$S$
is $i/(N-i+1)$, and since $i\le k \le (1/2 - \epsilon) N$, it follows that
\begin{equation}
\label{eq:estimate}
\frac{i}{N-i+1} \le \frac{(1/2 - \epsilon) N }{(1/2 + \epsilon)N + 1}.
\end{equation}
The right-hand side of~(\ref{eq:estimate})
is bounded by some constant strictly less than one.
Therefore $S$ is bounded by a convergent geometric series,
and this proves the lemma.
\end{proof}

\begin{lemma}
\label{lem:logNchoosek}
Assume that $k\le N/2$.
If $\log \binom{N}{k} \le N/2$, then $k\le N/4$.
\end{lemma}

As the proof below makes clear,
Lemma~\ref{lem:logNchoosek} remains true if
we replace ``$N/2$'' by ``$(1-\epsilon)N$,''
provided we replace ``$N/4$'' by a suitable constant times~$N$
and require that $N$ be sufficiently large.
We do not need this extra generality,
so we have stated Lemma~\ref{lem:logNchoosek} with specific constants
to make it easier to read.

\begin{proof}
If $k=0$ then the result is trivial, so assume that $k\ne0$.
Let $H(x) := - x \log x  - (1-x)\log (1-x)$ be the entropy function.
The basic reason why the lemma is true
is that $\log \binom{N}{k} \approx N \cdot H(k/N)$.
More precisely, by Stirling's approximation,
\begin{align*}
\log\binom{N}{k} & \ge N \cdot H(k/N) + \frac{1}{2} \log \frac{N}{k(N-k)}
  - \frac{1}{2} \log 2\pi 
   - \left(\frac{1}{12k} + \frac{1}{12(N-k)}\right)\log e \\
  & \ge N \cdot H(k/N) + \frac{1}{2} \log \frac{N}{k(N-k)} - 2.
\end{align*}
So if $\log \binom{N}{k} \le N/2$, then
\begin{equation*}
H(k/N) \le \frac{1}{2} - \frac{1}{2N} \log\frac{N}{k(N-k)} + \frac{2}{N}.
\end{equation*}
The expression $N/k(N-k)$ is minimized when $k=N/2$,
and by elementary calculus we find that
$(1/2x) \log (4/x)$ is minimized when $x = 4e$
(remember that in this paper, our logarithms are base~2).
Therefore, provided $N\ge 10$,
\begin{equation*}
H(k/N) \le \frac{1}{2} + \frac{\log e}{8e} + 0.2 \le 0.8.
\end{equation*}
It follows that if $N\ge 10$, $k/N \le H^{-1}(0.8) \le 1/4$ as desired.
If $N < 10$, then the lemma can be checked by direct computation.
\end{proof}

Now we are ready to prove a lower bound on~$\psi(n,g)$.

\begin{proposition}
\label{prop:psibound}
Let $\gamma:\mathbb{N} \to \mathbb{N}$ be a function
such that $\gamma(n) \le 2^{n-2}\!/n$
and such that $\gamma(n)$ outstrips $n\log n$.
Then for any fixed~$d$,
$\psigam{n} \ge n^d \psi(n, \gamma(n)/2)$
for all sufficiently large~$n$.
\end{proposition}

\begin{proof}
Let $N= 2^n$.
We are trying to find a lower bound on how many more
Boolean functions we can compute with $\gamma(n)$ gates
than we can compute with only $\gamma(n)/2$ gates.
Our main observation is that by using $O(n)$ extra gates,
we can change any single entry of the truth table of
any given Boolean circuit:
Simply use the $O(n)$ gates to test if the input equals
a specific $n$-bit value, and flip the output of the circuit if it does.

If $B$ denotes the set of truth tables of functions
computable with at most $\gamma(n)/2$ gates,
then our main observation implies that
if we are allowed up to $\gamma(n)/2 + O(n)$ gates, 
then at minimum we can also compute all the functions
on the \emph{boundary} $G(B)$ of~$B$,
i.e., the truth tables whose Hamming distance from~$B$ is~$1$.
We know very little about the structure of~$B$,
but we do have an estimate of its volume,
so we can obtain a lower bound on the size of its boundary
by appealing to a discrete isoperimetric inequality.
In particular, it follows from standard results\footnote{See
for example Bezrukov's survey paper~\cite{Bezrukov}.
Bezrukov states an isoperimetric inequality for
the \emph{inner} boundary~$\Gamma(B)$,
but this can be converted into an inequality for $G(B)$ as follows.
In the notation of Bezrukov's paper,
we may assume that $B$ is an optimal set~$L^N_m$ for some~$m$.
Then the radius-$(k+1)$ Hamming ball $S^N_{k+1}(0) \subseteq B \cup G(B)$,
so if we let $b = \left|B \cup G(B)\right|$, it follows that
as long as $k+1<N/2$,
$$ \left|\Gamma(B \cup G(B))\right| \ge \left|\Gamma\left(L^N_b\right)\right|
   \ge \left|\Gamma\left(S^N_{k+1}(0)\right)\right| = \binom{N}{k+1}.$$
On the other hand, $\Gamma(B \cup G(B)) \subseteq G(B)$ so
$\left|\Gamma(B \cup G(B))\right| \le \left|G(B)\right|$.}
that if we choose $k$ so that
\begin{equation}
\label{eq:B}
\sum_{i=0}^k \binom{N}{i} \le |B| < \sum_{i=0}^{k+1} \binom{N}{i},
\end{equation}
then $|G(B)| \ge \binom{N}{k+1}$.
We claim that there is some constant~$c$
such that $|B| < c\left|G(B)\right|$ for all large~$n$.
To see this, note that since $\gamma(n) \le N/4n$,
Proposition~\ref{prop:shannon} implies that for large~$n$,
\begin{equation*}
\log |B| \le \frac{N}{4n}\log\frac{N}{4n}
    + \left(\frac{N}{4n} + 4n\right)\log e
   \le \frac{2N}{4n} \log \frac{N}{4n}
    = \frac{N}{2}\left(1 - \frac{\log 4n}{n}\right) \le N/2.
\end{equation*}
But (\ref{eq:B}) yields the lower bound $|B| \ge \binom{N}{k}$,
so by Lemma~\ref{lem:logNchoosek}, $k\le N/4$.
This fact, together with the upper bound on~$|B|$ given by~(\ref{eq:B}),
implies (by Lemma~\ref{lem:largedev})
that $|B|$ is bounded by a constant times $\binom{N}{k+1}$.
Since $|G(B)| \ge \binom{N}{k+1}$, our claim is proved.

So when an additional $O(n)$ gates are allowed,
the number of computable functions is multiplied by
at least some constant factor~$K > 1$.
Now in fact we have $\gamma(n)/2$ additional gates at our disposal,
and $\gamma(n)/2$ outstrips $n\log n$,
so the multiplicative factor is greater than $K^{c\log n}$
for any constant~$c$,
and this eventually grows faster than $n^d$ for any fixed~$d$.
\end{proof}

\section{The main result}
\label{sec:main}

\begin{theorem}
\label{thm:main}
Assume that, for some $\epsilon>0$,
there exists a $2^{n^\epsilon}$-hard pseudorandom number generator~$G$
in $P/poly$.
Then there exists a quasi-polynomial function~$q$
and a nearly-linear-natural
property of density~$\Omega(2^{-q(n)})$
separating $NP$ from $P/poly$.
\end{theorem}

In fact, as will be apparent from the proof,
the property we exhibit contains functions
that are probably not $NP$-hard,
so our separation is actually stronger than $NP \not\subseteq P/poly$.

The main tool in our proof of Theorem \ref{thm:main}
is the following concept.

\begin{definition}
Given $\gamma:\mathbb{N}\to\mathbb{N}$, we define
a Boolean function $f$ on $n$ variables
to be \emph{$\gamma$-discriminating} if
either of the following two conditions holds:
\begin{enumerate}
\item $n$ is not a power of~$2$.
\item $n = 2^m$ for some $m$ and
  \begin{enumerate}
  \item $f(x) = 1$ for at least $2^n\!/n$ values of
        (the $n$-digit binary string)~$x$, and
  \item $f(x) = 0$ if $x$ is the truth table of a Boolean function
        on $m$ variables that is computable by a Boolean circuit with
        at most $\gamma(m)$ gates.
  \end{enumerate}
\end{enumerate}
\end{definition}

If we let $M^\gamma_n$ be
the set of all $\gamma$-discriminating Boolean functions
on $n$ variables, then $(M^\gamma_n)$
is a Boolean function property
that we shall call \emph{$\gamma$-discrimination.}

The following easy lemma shows that $\gamma$-discrimination
is constructive, and gives a lower bound on its density.

\begin{lemma}
\label{lem:discrim}
Let $\gamma:\mathbb{N}\to \mathbb{N}$
be a time-constructible function
satisfying $\gamma(m) \le 2^m\!/m$ for all~$m$.
Then $\gamma$-discrimination
is a nearly-linear-natural property
with density~$\Omega(2^{-\psigam{\log n}})$.
\end{lemma}

\begin{proof}
Let $n$ denote the number of variables of our Boolean functions.
If $n$ is not a power of~$2$ then the lemma is trivial,
so assume that $n = 2^m$.

First we note that,
since $\gamma(m) \le 2^m\!/m$,
it is easy to deduce from Proposition~\ref{prop:shannon} that
the number of Boolean circuits with $m$ inputs and at most $\gamma(m)$ gates
is much less than $2^{2^m} = 2^n$.

Let us check constructivity.
To verify that a given truth table is
the truth table of a $\gamma$-discriminating function,
we must check that
the fraction of entries equal to~$1$ is at least $1/n$,
and we must also check that the entries indexed by truth tables
of functions computable by circuits with at most $\gamma(m)$ gates
are~$0$.
Let $N = 2^n$ be the size of the truth table.
Counting~$1$'s clearly takes time that is nearly linear in~$N$,
but to check the forced~$0$'s
we must first compute $\gamma(m)$,
then run through each possible Boolean circuit in turn,
computing its $n$ truth table values,
and checking that the corresponding entry of the given truth table is~$0$.
Since $\gamma$ is time-constructible,
computing $\gamma(m)$ takes time $O(2^m)$,
so evaluating $\gamma$ at $m = \log \log N$
takes time at most polylogarithmic in~$N$.
Enumerating all the circuits is a straightforward process,
and the total number of circuits to be enumerated is at most~$N$,
so the entire computation
takes time at most $N$ multiplied by some factors
that are polylogarithmic in~$N$.

It remains to estimate the density.
If we were to ignore condition 2(a)
in the definition of a $\gamma$-discriminating function,
then we would simply be counting functions that
must be~$0$ in certain positions and are unrestricted otherwise,
so the total number of functions on $n$ variables
would be precisely $2^{2^n - \psigam{m}}$.
\relax From this we can get a lower bound for
the true number of $\gamma$-discriminating functions
by subtracting off the total number of Boolean functions on $n$ variables
whose truth tables have at most $2^n\!/n$ entries equal to~$1$.
This latter quantity is
\begin{equation*}
\sum_{i=0}^{2^n\!/n} \binom{2^n}{i}.
\end{equation*}
By Lemma~\ref{lem:largedev},
\begin{equation*}
\sum_{i=0}^{2^n\!/n} \binom{2^n}{i} = O\left(\binom{2^n}{2^n\!/n}\right)
  = 2^{O(2^n\log n)/n},
\end{equation*}
where the second equality is a routine application of Stirling's approximation.
It follows that for some constant~$c$,
the number of $\gamma$-discriminating functions is at least
\begin{align*}
&2^{2^n - \psigam{m}} - 2^{c(2^n\log n)/n} \\
&\quad = 2^{2^n} 2^{-\psigam{m}}
   (1 - 2^{c(2^n\log n)/n - 2^n + \psigam{m}}).
\end{align*}
Again, $\psigam{m}$ is vanishingly
small compared to $2^{2^m} = 2^n$, so the density is indeed
eventually lower-bounded by a constant times $2^{-\psigam{m}}$.
\end{proof}

We are now ready for the proof of our main result.

\begin{proof}[Proof of Theorem \ref{thm:main}]
By hypothesis, there exists $c\ge1$ such that
the pseudorandom number generator~$G_k$ is in $SIZE(k^c)$.
Choose some number $e>1 + c/\epsilon$,
and let $\gamma$ be the function $\gamma(m) = m^e$.
Then we claim that the desired property is simply $\gamma$-discrimination.

By Lemma~\ref{lem:discrim} we know that $\gamma$-discrimination
is nearly-linear-natural with density $\Omega(2^{-\psigam{\log n}})$.
Since $\gamma$ is a polynomial function, Proposition~\ref{prop:shannon}
implies that this density is indeed $\Omega(2^{-q(n)})$
for some quasi-polynomial~$q$.

We next show that $\gamma$-discrimination is quasi-useful against $P/poly$.
Given $f_n\in M_n^\gamma$, define the property $(C_m)$
by letting a function with truth table~$x$ be in $C_m$
if and only if $f_{2^m}(x) = 1$.
Since $f$ is a $\gamma$-discriminating function,
it follows that $(C_m)$ is useful against $SIZE(m^e)$
and that $(C_m)$ has density $\Omega(2^{-m})$.
Invoking Theorem~\ref{thm:RR} with $d=1$, we see that
$(C_m)$ cannot be $P/poly$-constructive.
In other words, $(f_n)\notin P/poly$,
which means that $\gamma$-discrimination is indeed quasi-useful
against $P/poly$.

Finally, let $(f_n)$ be the sequence of
$\gamma$-discriminating functions that are~$0$
only when forced to be by condition 2(b) and that are~$1$ otherwise.
Then $(f_n)$ is in $NP$,
in the sense that the language $L$ defined by
\begin{equation*}
x \in L \iff f_n(x)=0
\end{equation*}
is in $NP$.\footnote{Some authors might prefer to say that
$(f_n)$ is in co-$NP$, but since we could have chosen to interchange
the roles of $0$ and~$1$ in the definition of $\gamma$-discrimination,
this distinction is of no importance.}
The reason is that, for $n$ a power of~$2$,
a Boolean circuit with truth table~$x$
is a certificate for membership in~$L$,
and such a circuit has size $\gamma(\log n)$,
which is polynomial (even polylogarithmic) in~$n$, the size of~$x$.
This completes the proof.

Note that as we remarked earlier,
$(f_n)$ is almost certainly not $NP$-complete,
so that we have actually separated $P/poly$
from a subclass of~$NP$.
\end{proof}

\section{An unconditional result}
\label{sec:unconditional}

As we remarked in the introduction,
the idea behind the proof of Theorem~\ref{thm:main}
can be adapted to prove a non-uniform version of the result
that has no unproven hypotheses.
Now, it turns out that this unconditional result
can also be proven by a counting argument
that does not use any self-reference.
Since the two arguments are very different in flavor,
we present both of them below.

First we need a non-uniform version of Lemma~\ref{lem:discrim}.

\begin{lemma}
\label{lem:nonunif}
Let $\gamma:\mathbb{N}\to \mathbb{N}$
be a function satisfying $\gamma(m) \le 2^m\!/m$ for all~$m$.
Then $\gamma$-discrimination is a non-uniformly linear-natural property
with density~$\Omega(2^{-\psigam{\log n}})$.
\end{lemma}

When we say ``non-uniformly linear-natural property,''
we of course mean that membership can be decided by
circuits whose size is linear in the size of the truth table.

\begin{proof}
The proof is the same as the proof of Lemma~\ref{lem:discrim}
except when it comes to $\Gamma$-constructivity.

Let $n = 2^m$ denote the number of variables of our Boolean functions.
As we said before, to verify that a given truth table is
the truth table of a $\gamma$-discriminating function,
we must check that
the fraction of entries equal to~$1$ is at least $1/n$,
and we must also check that the entries indexed by truth tables
of functions computable by circuits with at most $\gamma(m)$ gates
are~$0$.
Let $N = 2^n$ be the size of the truth table.
We can count the number of~$1$'s using $O(N)$ gates,
for example by using carry-save addition~\cite{Paterson-Pippenger-Zwick}.
Also, for each $n$, the set of truth table entries that must be~$0$
is fixed, so this condition can be checked using a number of gates
that is proportional to the number of forced~$0$'s
(even if $\gamma$ is not time-constructible);
this number is certainly $O(N)$.
\end{proof}

\begin{theorem}
\label{thm:nonunif}
Let $\gamma, \lambda : \mathbb{N} \to \mathbb{N}$
be functions such that
$\gamma$ outstrips $\lambda$
and such that
$m \le \gamma(m) \le 2^m\!/m$ for all~$m$.
Let $\Gamma = SIZE(\gamma)$ and
let $\Lambda = SIZE(\lambda)$.
Then there exists a $\Gamma$-natural property~$(C_n)$
with density $\Omega(2^{-\psigam{\log n}})$
that is quasi-useful against $\Lambda$.
\end{theorem}

\begin{proof}
We argue by contradiction.
Assume, as a reductio hypothesis, that
there is no $\Gamma$-natural property~$(C_m)$
with density $\Omega(2^{-\psigam{\log m}})$
that is quasi-useful against $\Lambda$.
Then we claim that $\gamma$-discrimination is
quasi-useful against~$\Lambda$.

To see this, pick an arbitrary sequence of functions
$f_n \in M^\gamma_n$.
Define a property $(C_m)$ by letting a function of $m$ variables
with truth table~$x$ be in~$C_m$ if and only if $f_{2^m}(x)=1$.
Then by condition~2(a)
in the definition of a $\gamma$-discriminating function,
$(C_m)$ has density $\Omega(2^{-m})$.
By assumption, $\gamma(\log m) \ge \log m$, and
it is easy to see that there are more than $m$ distinct Boolean
functions computable with $\log m$ gates and $\log m$ inputs,
so the density of~$(C_m)$ is $\Omega(2^{-\psigam{\log m}})$.
By condition~2(b),
if $g_m \in C_m$ is any sequence of Boolean functions,
then the minimum circuit size of~$g_m$ exceeds~$\gamma(m)$,
and hence $(g_m)\notin\Lambda$ since $\gamma$ outstrips~$\lambda$.
In other words, $(C_m)$ is quasi-useful
(in fact, useful) against~$\Lambda$.
Therefore, by our reductio hypothesis,
membership in $(C_m)$ is not $\Gamma$-computable.
It follows that $(f_n)\notin\Gamma$,
and a fortiori $(f_n)\notin\Lambda$.
Therefore $(f_n)$ is quasi-useful against~$\Lambda$, as claimed.

But since $n \le \gamma(n) \le 2^n\!/n$,
Lemma~\ref{lem:nonunif} tells us that $(M^\lambda_n)$
is $\Gamma$-natural with density $\Omega(2^{-\psigam{\log n}})$.
Combined with the quasi-usefulness against~$\Lambda$ that we just proved,
this fact contradicts our reductio hypothesis,
so the theorem is proved.
\end{proof}

Observe that a curious feature of the above proof is that
it is highly ineffective.
The natural property whose existence is asserted is
not explicitly exhibited,
nor can an explicit example be extracted from the proof,
which is intrinsically a proof by contradiction.
Note also that a $SIZE(\gamma)$-natural property
is not necessarily ``constructive'' in the intuitive sense
even if $\gamma$ is polynomial, because $SIZE(\gamma)$
is a \emph{non-uniform} complexity class.
Nevertheless, we feel that Theorem~\ref{thm:nonunif}
remains of some interest because
it is an unconditional result.

Next we present the promised counting argument,
which in fact yields a stronger result than Theorem~\ref{thm:nonunif}.

\begin{theorem}
\label{thm:counting}
Let $\gamma, \lambda : \mathbb{N} \to \mathbb{N}$
be functions such that $\lambda(n) = \Omega(n\log n)$,
$\gamma$ outstrips $\lambda$,
and $\gamma(n) \le 2^{n-2}\!/n$ for all~$n$.
Let $\Lambda = SIZE(\lambda)$.
Then there exists a non-uniformly linear-natural property
with density at least $1/\psigam{n}$
that is useful against~$\Lambda$.
\end{theorem}

\begin{proof}
Let us first give a somewhat informal proof
that conveys the essential idea.  Let $N=2^n$.

As usual, think of Boolean functions on $n$ variables
as represented by their truth tables.
Let $G_n$ be the set of Boolean functions on $n$ variables computable
by circuits of size~$\gamma(n)/2$.
For each $g\in G_n$,
imagine a Hamming ball of volume $2^N\!/\psigam{n}$
centered at~$g$
(by a \emph{Hamming ball centered at~$g$} we mean the set of
all Boolean functions within a certain Hamming distance from~$g$).
There are $\psi(n, \gamma(n)/2) < \psigam{n}$ such balls,
so the total volume of these balls is less than $2^N$.
Therefore there must exist a function~$f_n$ outside all of these balls.
It follows that there is a Hamming ball~$B_n$ of volume
$2^N\!/\psigam{n}$ around~$f_n$
that is disjoint from~$G_n$.
Then since $\gamma$ outstrips $\lambda$,
$(B_n)$ is a property that is useful against~$\Lambda$.
Its density is $1/\psigam{n}$.
Moreover, testing for membership in $B_n$ amounts to
computing Hamming distance from~$f_n$,
which can be done with circuits of linear size.

This completes the informal proof.
The only point that is not entirely rigorous is
the assumption that there exists a Hamming ball
whose volume is exactly $2^N\!/\psigam{n}$;
this may not be true because
the volume of a Hamming ball is necessarily
a sum of consecutive binomial coefficients.
For a rigorous argument, we choose our Hamming balls to have radius~$r$,
where $r$ is chosen so that
\begin{equation}
\label{eq:hamming}
\sum_{i=0}^{r-1} \binom{N}{i} < \frac{2^N}{\psi_n}
  \le \sum_{i=0}^r \binom{N}{i},
\end{equation}
where we have abbreviated $\psigam{n}$ to $\psi_n$ to ease notation.
Then the property of being in~$B_n$ certainly
has density at least $1/\psi_n$,
so all that needs to be checked is that $\psi(n, \gamma(n)/2)$
such Hamming balls have total volume strictly less than~$2^N$,
i.e., that
\begin{equation}
\label{eq:totalvolume}
\psi(n, \gamma(n)/2) \sum_{i=0}^r \binom{N}{r} < 2^N.
\end{equation}
To prove this,
observe that
we just need to show that the ratio
$\biggl(\sum_{i=0}^r \binom{N}{i}\biggr)\bigg/
 \biggl(\sum_{i=0}^{r-1} \binom{N}{i}\biggr)$
is bounded by a polynomial function of~$n$,
because then (\ref{eq:totalvolume})
will follow from (\ref{eq:hamming}) and
Proposition~\ref{prop:psibound}.  Now
\begin{equation*}
\frac{\sum_{i=0}^r \binom{N}{i}}{\sum_{i=0}^{r-1} \binom{N}{i}}
 =  \frac{\binom{N}{r}} {\sum_{i=0}^{r-1} \binom{N}{i}} + 1
 \le \frac{\binom{N}{r}}{\binom{N}{r-1}} + 1
 = \frac{N+1}{r}.
\end{equation*}
So we are reduced to showing that $(N+1)/r$ is bounded by a
polynomial function of~$n$.
To prove this, remember that by assumption $\gamma(n) < N/n$,
so Proposition~\ref{prop:shannon} implies that
$\psi_n \le (N/n)^{N/n}e^{N/n+4n}$.
Taking logarithms and dividing by~$N$, we deduce that
\begin{equation*}
\frac{\log \psi_n}{N}
 \le \frac{1}{n} \log\frac{N}{n} + \left(\frac{1}{n}+\frac{4n}{N}\right)\log e
  = 1 - \frac{\log n}{n} + \left(\frac{1}{n} + \frac{4n}{N}\right)\log e.
\end{equation*}
The $(\log n)/n$ term in this expression dominates, so for large~$n$,
\begin{equation}
\label{eq:largen}
1 - \frac{\log \psi_n}{N} \ge \frac{\log n}{2n}.
\end{equation}
On the other hand, from~(\ref{eq:hamming}) we have
\begin{equation*}
\frac{2^N}{\psi_n} \le \sum_{i=0}^r \binom{N}{i}
   \le \sum_{i=0}^r N^i = \frac{N^{r+1} - 1}{N-1} \le N^{r+1}.
\end{equation*}
Taking logarithms, we get $N - \log \psi_n \le (r+1)\log N$,
which combined with (\ref{eq:largen}) implies that for large~$n$,
\begin{equation*}
\frac{r+1}{N} \ge \frac{1 - (\log \psi_n)/N}{\log N}
    \ge \frac{\log n}{2n^2}.
\end{equation*}
We are now done, because $N/(r+1)$ and $(N+1)/r$
are within a constant factor of each other.
\end{proof}

\section{Vadhan's variation}
\label{sec:vadhan}

After my talk at FOCS 2008, Salil Vadhan showed me a
different way to prove the main result of this paper,
assuming only that SAT is hard
(and therefore, of course, separating $P/poly$ from $NP$
but not from a subclass of~$NP$).
With his kind permission,
I include his argument here.

\begin{theorem}
Assume that SAT is not computable by circuits of size $2^{n^\epsilon}$.
Let $\gamma$ be a function such that $\gamma(n)$ outstrips
$(\log n)^{1/\epsilon}$.
and let $q(n) = 2^{\gamma(n)}$.
Then there exists a sublinear-natural property of density~$2^{-q(n)}$
that separates $NP$ from $P/poly$.
\end{theorem}

\begin{proof}
To ease notation, let $m = \gamma(n)$.
Fix some way of encoding SAT instances as binary strings.
Let $C_n$ comprise all Boolean functions~$f$ of $n$ variables
with the following property.
If the last $n-m$ bits of~$x$ are all zero,
then $f(x)$ is $1$ or~$0$ according to whether or not
the first $m$ bits of~$x$ encode a satisfiable instance of SAT.
(If any of the last $n-m$ bits of~$x$ are nonzero,
then $f(x)$ can be anything.)
Then $C_n$ has density $1/2^{2^m} = 2^{-q(n)}$.
By our assumption on the hardness of SAT,
functions in $C_n$ cannot be computed by circuits of size $2^{m^\epsilon}$.
Since $m^\epsilon$ grows faster than $d\log n$ for any fixed~$d$,
this shows that $(C_n)$ is useful against $P/poly$.
Checking membership in~$C_n$ can be done in time $poly(m) \cdot 2^m$,
which is certainly sublinear in~$2^n$.
\end{proof}

\section{Final remarks}

It is natural to ask if our results give any new hope
for proving strong circuit lower bounds.\footnote{For
a survey of other possible approaches to breaking
the naturalization barrier, see~\cite{Allender}.}
It is probably difficult to prove unconditionally that,
say, $n^{\log n}$-discrimination is useful against
a strong complexity class~$\Lambda$,
not only because that would separate $NP$ from~$\Lambda$,
but also because $\gamma$-discrimination is closely
related to the circuit minimization problem,
whose complexity is known to be difficult to get a handle on;
see~\cite{Kabanets-Cai}.

However, even as a \emph{potential} candidate
for an almost-natural proof of $NP \not\subseteq P/poly$,
$\gamma$-discrimination has an illuminating feature.
Namely, the only thing that prevents a $\gamma$-discriminating
function from looking like a random function is
the presence of certain forced~$0$'s in the truth table.
Moreover, the proportion of forced~$0$'s goes to zero
fairly rapidly as $n$ goes to infinity.
This illustrates the fact that largeness can be destroyed
by imposing what seems intuitively to be a relatively small
amount of ``structure'' on a random function.
Therefore, the intuition that there is some constructive property
of random functions that suffices to prove strong circuit lower bounds
is not completely destroyed by the Razborov--Rudich results;
a minor alteration of a random property may still work.

It is also worth noting that existing circuit lower bound proofs
might still be mined for ideas to break the naturalization barrier.
Some linear lower bounds, such as those of Blum~\cite{Blum} and
Lachish and Raz~\cite{Lachish-Raz}, do not relativize
and are not known to naturalize.
Even proofs that are known to naturalize are not necessarily devoid
of useful ideas.
For example, in the course of analyzing a proof by Smolensky,
Razborov and Rudich identify three properties
$C_1 \subseteq C_2 \subseteq C_3$ that are implicit in the proof,
and show that $C_2$, and a~fortiori~$C_3$, are natural.
However, $C_1$ is constructive but not known to be large,
so it is conceivable (though admittedly unlikely) that
$C_1$ is only \emph{almost} large and is actually useful.
Of course, one would still need to identify and use
some feature of~$C_1$ that is not shared by~$C_2$
in order to prove a stronger circuit lower bound than Smolensky's,
but the point is that the usefulness of~$C_1$ is not \emph{automatically}
ruled out by the fact that Smolensky's argument naturalizes.
In theory, it could still be fruitful to study~$C_1$.

Finally, recall that as evidence that largeness is hard to circumvent,
Razborov and Rudich showed that any formal complexity measure
automatically yields a large property.
Knowing that almost-natural proofs exist,
we could perhaps try to come up with something that is
almost, but not quite, a formal complexity measure.
Unfortunately, as of now, this tempting idea remains purely speculative.

\section{Acknowledgments}

I would like to thank Bob Beals, Steve Boyack, Sandy Kutin,
and Avi Wigderson for helpful discussions and encouragement.



\end{document}